\documentclass[runningheads]{llncs}

\usepackage{amsmath}
\usepackage{amssymb}
\usepackage{graphics}
\usepackage{latexsym}
\usepackage{enumerate}
\usepackage[normalem]{ulem}
\usepackage{color}
\usepackage{float}
\usepackage{graphicx}
\usepackage[loose]{subfigure}
\usepackage{cite}
\usepackage{epstopdf}
\usepackage{appendix}
\usepackage{placeins}

\usepackage[ruled,vlined,linesnumbered]{algorithm2e}
\DontPrintSemicolon
\usepackage[pagewise]{lineno}

\begin{document}

\newcommand{\edge}[2]{(#1,#2)}
\newcommand{\seppair}[2]{\langle#1,#2\rangle}

\title{Straight-line  Drawings of 1-Planar Graphs\thanks{Supported
by Deutsche Forschungsgemeinschaft (DFG) Br835/20-1}}
%
%
\author{Franz J. Brandenburg\inst{1}} 
\authorrunning{Franz J. Brandenburg}
\institute{94030 Passau, Germany
\email{brandenb@informatik.uni-passau.de}}
\maketitle              

\begin{abstract}
 A graph is 1-planar if it can be drawn in the plane so that each
edge is crossed at most once. However, there are 1-planar graphs
which do not admit a straight-line 1-planar drawing. We show that
every 1-planar graph has a straight-line   drawing with a
two-coloring of the edges, so that edges of the same color do not
cross. Hence,
 1-planar graphs have geometric thickness two. In
addition,   each edge is crossed by edges with a common vertex if it
is crossed more than twice. The drawings use  high precision
arithmetic with numbers with $O(n \log n)$ digits and can be
computed in linear time from a
  1-planar drawing.
\end{abstract}

\section{Introduction}

Straight-line drawings of graphs, also known as   rectilinear or
geometric drawings, are an important topic in Graph Drawing, Graph
Theory, and  Computational Geometry. The existence of straight-line
drawings of planar graphs was discovered   by   Steinitz and
Rademacher \cite{sr-34}, Wagner \cite{w-bv-36}, F\'{a}ry
\cite{fary-48}, and Stein \cite{s-cm-51}. Hence, straight-line is no
restriction for planar graphs. The first algorithms for constructing
straight-line planar drawings need high precision arithmetic
\cite{con-dpgn-8585, cyn-lacdpg-84, t-convex-60, t-hdg-63}. Later,
de Fraysseix et al.~\cite{fpp-hdpgg-90} and Schnyder
\cite{s-epgg-90} showed  that planar graphs can be drawn
straight-line   on a grid of quadratic size. The drawings can be
convex, so that the faces are convex polygons if the graphs are
3-connected \cite{bfm-cd3cpg-07, ck-cgd-97, k-dpguco-96}. However,
the  produced drawings are not aesthetically pleasing,  since they
have a low angular resolution.

 A drawing (or embedding) of  a graph in the plane   is 1-\emph{planar} if each edge is crossed at most once.
 A graph is 1-planar if it admits such a drawing.   Here, straight-line is a real restriction.
 Thomassen
\cite{t-rdg-88} showed that a 1-planar drawing can be transformed
into a straight-line 1-planar drawing if and only if it does not
contain a B- or a W-configuration, see Fig.~\ref{fig:configs}. This
fact was rediscovered by Hong et.~al~\cite{help-ft1pg-12}. B- and
W-configurations are related to separation pairs, so that there is a
pair of crossed edges in the outer face of a component.
 Didimo \cite{d-ds1pgd-13} 
 showed that straight-line 1-planar drawings of
$n$-vertex graphs have at most $4n-9$ edges, whereas there are
1-planar graphs with $4n-8$ edges 
\cite{bsw-1og-84}.

\begin{figure}[t]
\centering
\subfigure[ ] {    
     \includegraphics[scale=0.55]{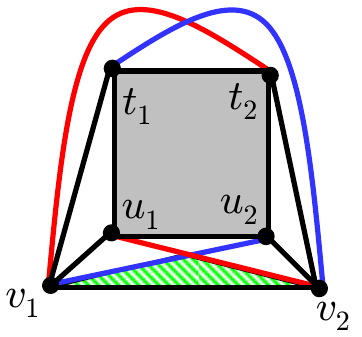}  
      \label{fig:Wconf}
  }
  \hspace{2mm}
\subfigure[ ] {    
    \includegraphics[scale=0.6]{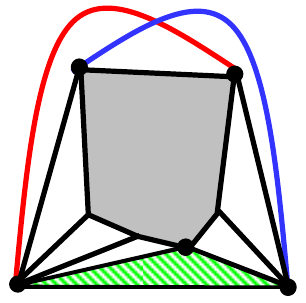}  %
      \label{fig:Bconf}
  }
  \hspace{2mm}
  \subfigure[ ] {    
    \includegraphics[scale=0.6]{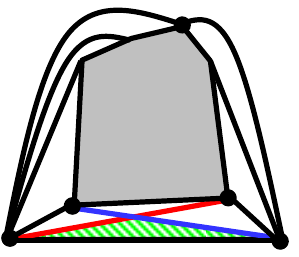}  %
      \label{fig:Bconf-reverse}
  }
 \caption{(a) A W-configuration with base $\edge{v_1}{v_2}$, top vertices
 $t_1$ and $t_2$ and a
 first inner face drawn green with a texture, (b)
 a B-configuration with a pair of crossed edges in the outer face,
 which in (c) is removed
 by a flip of the component that changes the embedding.
  }
  \label{fig:configs}
\end{figure}

 Straight-line is also a real restriction for \emph{fan-crossing}
\cite{b-fan-20},  \emph{fan-crossing free}
\cite{b-fcf-21,cpkk-fan-15}, and 2-\emph{planar} graphs
\cite{bkr-optimal2-17}, which each generalize 1-planar graphs.
Fan-crossing graphs admit drawings with the crossing of an edge by
edges of a \emph{fan}, that is the crossing edges   are incident to
a common vertex, whereas such crossings are excluded for
fan-crossing free graphs \cite{b-FOL-18}. Then only crossings by
\emph{independent  edges} are allowed, that is the edges have
distinct vertices. Note that there are graphs that are fan-crossing
and fan-crossing free, but not 1-planar \cite{b-fan-fcf-18}.
Straight-line fan-crossing drawings are \emph{fan-planar}
\cite{ku-dfang-14}, which are fan-crossing and exclude crossings of
an edge from both sides. However, there are fan-crossing graphs that
are not fan-planar \cite{b-fan-20}.   A fan-crossing free drawing of
an $n$-vertex graph with $4n-8$ edges is 1-planar
\cite{cpkk-fan-15}, which cannot be drawn straight-line
\cite{d-ds1pgd-13}. Finally, the crossed dodecahedron graph is
2-planar and fan-crossing, but it does not admit a straight-line
2-planar drawing, since it has a unique 2-planar embedding
\cite{bkr-optimal2-17}, which is shown in
Fig.~\ref{fig:Xdodecaeder}.

\begin{figure}[h]
\centering
\subfigure[ ] {    
     \includegraphics[scale=0.35]{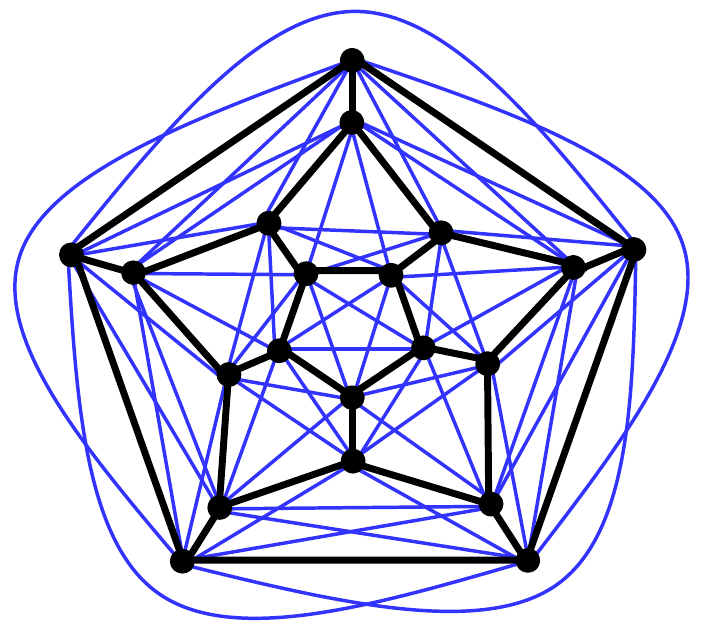}  
      \label{fig:Xdodecaeder}
  }
  \hspace{12mm}
\subfigure[ ] {    
    \includegraphics[scale=0.3]{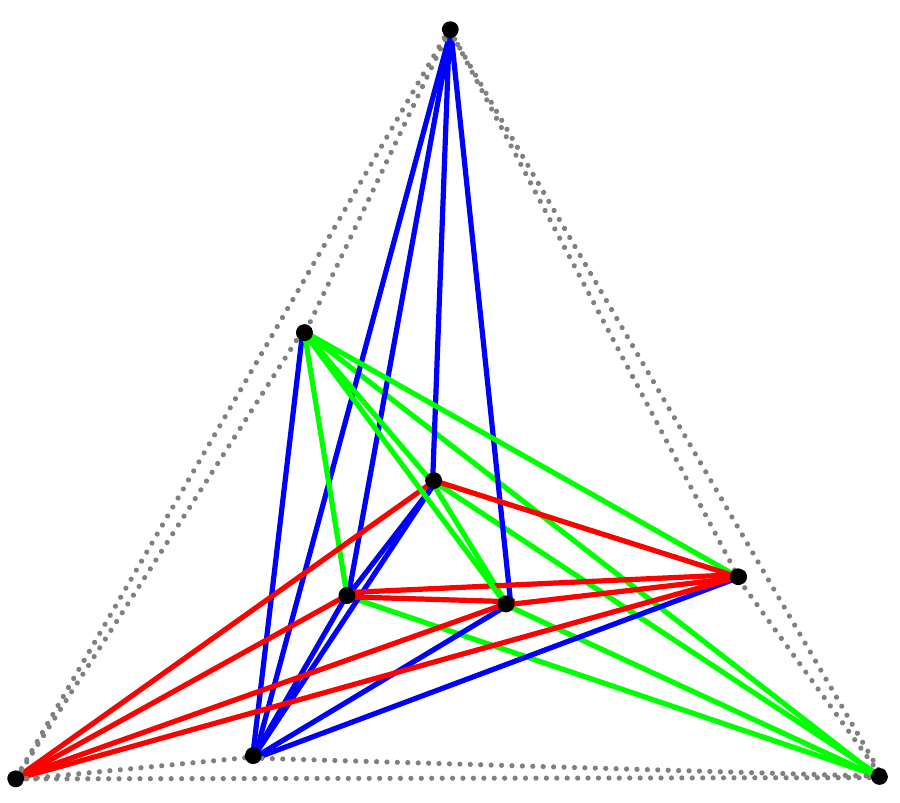}  
      \label{fig:K9-thick3-a}
  }
 \caption{(a) The crossed dodecahedron graph, which does not admit
 a straight-line  2-planar drawing and (b) $K_9$, which  has
 rectilinear thickness three. Dotted edges are uncrossed and can be
 colored arbitrarily.
  }
  \label{fig:examples}
\end{figure}

  A set of $k + \ell$ edges of a graph is said to form a $(k,
\ell)$-\emph{grid}  \cite{afps-grids-14, ppst-tgnlg-05} if each of
the first $k$ edges crosses all of the remaining $\ell$ edges,   see
Fig.~\ref{fig:grids}. A $(k, \ell)$-grid is \emph{radial} if, in
addition, the first $k$ edges are incident to the same vertex, that
is they form a fan, and \emph{natural} if all $k+\ell$ edges are
independent. Hence, a graph is $k$-planar, fan-crossing and
fan-crossing free, respectively, if and only if it avoids $(k+1,1)$,
natural $(2,1)$ and  radial $(2,1)$-grids, respectively. We call a
graph \emph{tri-fan-crossing} if it admits a drawing so that
 each edge is crossed by edges of a fan
 if the edge is crossed at least three times.

\begin{figure}[t]
\centering
\subfigure[ ] {    
     \includegraphics[scale=0.45]{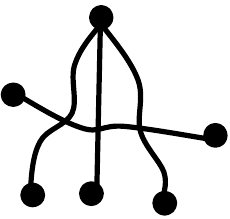}  
      \label{fig:grid-fan}
  }
  \hspace{4mm}
\subfigure[ ] {    
    \includegraphics[scale=0.45]{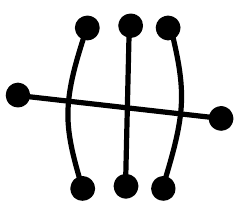}  
      \label{fig:grid-fcf}
  }
  \hspace{4mm}
  \subfigure[ ] {    
    \includegraphics[scale=0.45]{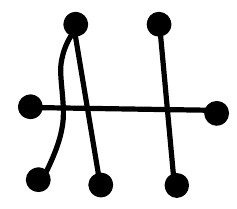}  
      \label{fig:grid-mix}
  }
    \hspace{4mm}
\subfigure[ ] {    
    \includegraphics[scale=0.45]{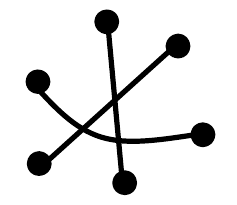}  
      \label{fig:quasi}
  }
 \caption{(a) A radial $(3,1)$-grid or a fan-crossing, (b) a natural $(3,1)$-grid
 and a fan-crossing free drawing, (c) a $(3,1)$-grid, and (d)  three
 pairwise crossing edges.
  }
  \label{fig:grids}
\end{figure}

 The \emph{thickness} of a graph 
\cite{t-thickness-63}  is the minimum number of planar graphs into
which the edges can be partitioned. Thickness has received much
attention \cite{mos-thickness-98}.  It has   applications in VLSI
design \cite{aks-mgev-91}, where crossed edges must be embedded in
different layers, and in graph visualization, where there is an edge
coloring so that edges of the same color do not cross. The thickness
of a graph is bounded by its \emph{arboricity}, which is the minimum
number of forests into which a graph can be decomposed. The
arboricity of a graph is the maximum density of any subgraph, that
is $\max \{\lceil \frac{|E[U]|}{|U|-1}\rceil\}$, where $U$ is a
subset of vertices and $E[U]$ is the set of edges with both vertices
in $U$
 \cite{nash-arboricity-61}.  It can be computed in $O(n^{3/2})$ time if graphs have $O(n)$ many edges
 \cite{gw-ffg-92}.
 One third of the arboricity is a lower bound for
 the thickness. 

The \emph{rectilinear thickness} of a graph G, also known as
\emph{real linear thickness} \cite{k-tcg-73} or \emph{geometric
thickness} \cite{e-stgt-02}, is the minimum number of colors in a
straight-line drawing of $G$, so that edges of the same color do not
cross. It is a real restriction, since
 complete graphs  have thickness $\lfloor (n+7)/6
 \rfloor$ for $n \neq 9,10$ \cite{ag-thickness-76},
  whereas they have geometric thickness at least
$\lceil (n+1)/5.646 \rceil$ \cite{deh-gtcg-00}. Moreover,
Eppstein~\cite{e-stgt-02} showed that for every $k$, there are
graphs with thickness three and rectilinear thickness at least $k$.
Hence, geometric thickness is not bounded by thickness.

Graphs with thickness two, also called \emph{biplanar} graphs
\cite{harary-61}, were studied by Hutchinson et
al.~\cite{hsv-rstg-99}. They showed that graphs with rectilinear
thickness two have at most $6n-18$ edges, and that the complete
graph $K_8$ is the largest complete graph with rectilinear thickness
two,    since $K_9$ and $K_{10}$ have thickness three
\cite{ag-thickness-76}. Note that $K_9$ has rectilinear thickness
three, as shown in Fig.~\ref{fig:K9-thick3-a}.

Geometric thickness specializes to \emph{book thickness} if all
vertices are placed in convex position.  It is known that planar
graphs have book thickness four \cite{y-epg4p-89}, where the lower
bound was proved just recently \cite{bkkpru-4pages-20,y-4pages-20},
and that the book thickness of 1-planar graphs is bounded by a
constant \cite{bbkr-book1p-17}.

In the remainder of this work, we introduce basic notions in Section
\ref{sect:prelim} and recall   methods for straight-line drawings of
planar graphs. In Section \ref{sect:main}, we first show that every
3-connected 1-planar graph admits a straight-line biplanar grid
drawing. Such drawings scaled down for the general case so that they
fit into the first inner face.

\section{Preliminaries} \label{sect:prelim}

We consider simple undirected graphs with $n$ vertices and colored
edges and assume that graphs are biconnected  and are given with a
drawing (or a simple topological  embedding) in the plane. For
convenience, we do not distinguish between a graph  and its 1-planar
drawing, a vertex,  its point in a drawing, and its position in a
canonical ordering, and  an edge and its line in a drawing. We also
refer to a leftmost lower neighbor or a lower (left) subgraph if
this is clear from the drawing of a planar graph.

Straight-line grid drawings of planar graphs can be constructed by
using the canonical ordering and the shift method, introduced by de
Fraysseix et al.~\cite{fpp-hdpgg-90}. Alternatively, Schnyder
realizers can be used \cite{s-epgg-90}. A \emph{canonical ordering}
of a triconnected planar graph $G$ is a bucket order $V_1, \ldots,
V_m$ of its vertices, so that further properties are satisfied
\cite{k-dpguco-96}. A bucket $V_i$ consists of a single vertex or a
set of vertices forming a path that is a part of the boundary of a
face of a planar drawing of $G$. The first bucket consists of two
vertices $v_1$ and $v_2$ in the outer face, called the \emph{base}.
The last bucket consists of another vertex in the outer face. Let
$G_k$ denote the subgraph induced by  the vertices from the first
$k$ buckets. Its outer face, called \emph{contour}, contains $v_1,
v_2$ and the vertices of $V_i$. Every $G_k$ is biconnected and
internally triconnected, so that separation pairs are on the
contour.

A canonical ordering of a triconnected planar graph can be computed
in linear time \cite{k-dpguco-96}. In general, a planar graph has
many canonical orderings, for example a leftmost and a rightmost one
\cite{k-dpguco-96}. In particular, if the base  is fixed, then  any
other vertex in the outer face can be chosen as the last vertex. If
the last vertex $t_1$ is in a separating triangle $(t_1, t_2, t_3)$,
then there are canonical orderings $t_2 <t_3 < w < t_1$ and $t_3 <
t_2 < w < t_1$, where vertex $w$ is in the interior of $(t_1, t_2,
t_3)$,  since there is a canonical ordering of a subgraph with $t_1$
and the vertices in the interior of $(t_1, t_2,t_3)$ removed and
with $t_2$ or $t_3$ as last vertex. The vertices in the interior of
$(t_1, t_2,t_3)$  follow $t_2$ and $t_3$ by biconnectivity. \\

 The drawing of a graph in the plane is 1-\emph{planar} if each edge is crossed
at most once. Then edges cross in pairs. The crossing point
partitions a crossed edge into two uncrossed \emph{segments},
whereas an uncrossed edge consists of a  single segment. A graph $G$
is 1-\emph{planar} if it admits a 1-planar drawing. A 1-planar
drawing is \emph{planar-maximal} if no uncrossed edge can be added
without violating 1-planarity. This property is assumed from now on.
Alternatively, one may consider \emph{maximal} 1-planar drawings,
which do not  admit the addition of any edge without violation. The
\emph{planar skeleton} 
is obtained by removing all pairs of crossed edges.   For an
algorithmic treatment, we use the \emph{planarization}, which is
obtained from a 1-planar drawing by treating each crossing point as
a special vertex of degree four \cite{help-ft1pg-12}. Clearly, a
1-planar drawing can be augmented to a planar-maximal   one in
linear time.

A \emph{separation pair} $\seppair{u}{v}$ partitions a graph into
components, so that $G-\{u,v\} = H_0^-,\ldots, H_r^-$ for some
$r\geq 1$. Let $H_i$ denote the subgraph induced by the vertices of
$H_i^-$ and the vertices from the separation pair, so that the
common edge $\edge{u}{v}$ is in the outer face.   The (edges between
$u$ and vertices of the) components $H_0,\ldots, H_k$
  are ordered clockwise at $u$ and counter clockwise at $v$.
Two consecutive components are separated by one or two pairs of
crossed edges if $G$ is planar-maximal 1-planar. The \emph{outer
component} $H_0$  contains vertices in the outer face if all crossed
edges are removed. All other  components are called \emph{inner
components}. Every inner component $H_i$ has a \emph{first inner
face}, which is the face next to edge $\edge{u}{v}$ in the subgraph
induced by $H_0,\ldots, H_i$, that is the later inner components
$H_{i+1}, \ldots H_r$ are removed.

There is a B- or a W-\emph{configuration}  if a pair of edges
crosses in the outer face  of a component \cite{t-rdg-88}. Alam et
al.~\cite{abk-sld3c-13} observed that  B-configurations can be
avoided if the embedding is changed by a flip, see
Figs.~\ref{fig:Bconf} and \ref{fig:Bconf-reverse}. A
\emph{W-configuration} consists of six vertices $v_1,v_2, u_1, u_2,
t_1, t_2$, so that there are two pairs of crossed edges
$\edge{v_1}{u_2}$, $\edge{v_2}{u_1}$ and $\edge{v_1}{t_2}$,
$\edge{v_2}{t_1}$, see Fig.~\ref{fig:Wconf}. Vertices $v_1$ and
$v_2$ are  the \emph{base} with  edge $\edge{u}{v}$ uncrossed and
vertices $t_1$ and $t_2$ are on \emph{top}. By a flip, the pairs
$(t_1, t_2)$ and $(u_1, u_2)$ can be exchanged, so that $u_1$ and
$u_2$ are on top in another embedding.\\

The drawing   of an edge-colored 1-planar graph  is called
\emph{specialized} if there are vertices $v_1, v_2, t_1$ and $t_2$,
so that the outer face is an isosceles triangle  $(v_1, v_2, t_2)$
and  edges incident to $v_2$ are  black. Moreover,  the drawing is
1-planar, except for
  edge $\edge{v_1}{t_1}$, which is crossed by edges incident to
  $t_2$ and   edges crossing $\edge{v_1}{t_1}$ are crossed at most
  twice.

  Hence, a specialized drawing is tri-fan-crossing. It  is
  1-planar if either edge $\edge{v_1}{t_1}$ or vertex $t_2$ is removed,
    see Figs.~\ref{fig:fpp-2} and \ref{fig:gamma-prime}.

Alam et al.~\cite{abk-sld3c-13} have shown that every triconnected
1-planar graph can be drawn straight line on a grid of quadratic
size,  except, possibly, for one edge in the outer face.
  They use the algorithm by Chroback and Kant \cite{ck-cgd-97}
  for a drawing of the planar skeleton, whereas we prefer the one
   by Kant \cite{k-dpguco-96}, which
  creates edges with slope $-1, 0$ and $+1$ on the contour.

\begin{lemma} \label{lem:draw-3conn}
A triconnected 1-planar graph $G$ admits a straight-line 1-planar
drawing if and only if $G$ has at most $4n-9$ edges. If existing,
there is a specialized  drawing on a grid of size at most $(4n-8)
\times (2n-4)$, so that the outer face is an isosceles right-angled
triangle. An edge is uncrossed in the straight-line drawing if and
only if it is uncrossed in (the 1-planar drawing of) $G$.   Its
length $\ell$ is bounded by $1 \leq \ell \leq 4n-8$.
\end{lemma}

\begin{proof}
Form \cite{bsw-bs-83, pt-gdfce-97} and \cite{d-ds1pgd-13} we obtain
that a 1-planar drawing of a planar-maximal 1-planar graph $G$ has a
  triangular face with three uncrossed edges if and only if $G$ has at most $4n-9$
edges. Choose this triangle as the outer face. As shown by Kant
\cite{k-dpguco-96}, there is a convex drawing of the planar skeleton
on a  grid of size at most $(2n-4) \times (n-2)$, so that the outer
face is an isosceles right-angled triangle. The drawing can be
turned into a strictly convex drawing on a grid of size at most
$(4n-8) \times (2n-4)$  if two extra shifts are used for any two
edges of a quadrangle that are in a line, as   shown in
\cite{abk-sld3c-13, br-scdpg-06}. See Fig.~\ref{fig:faces} for an
illustration. Every pair of crossed edges is reinserted into the
quadrangle from which it was removed. Uncrossed edges are colored
black. There is a black and  a red edge if two edges cross, where
the coloring is chosen so that all edges incident to   vertex $v_2$
are black. Then the drawing is specialized. At its creation, an
uncrossed edge has length at least one, since there is a grid
drawing. Later on, edges are stretched by horizontal shifts. The
base is the longest edge and is a horizontal line of length at most
$4n-8$.
\end{proof}

\begin{figure}[t]
\centering
\subfigure[ ] {    
     \includegraphics[scale=0.3]{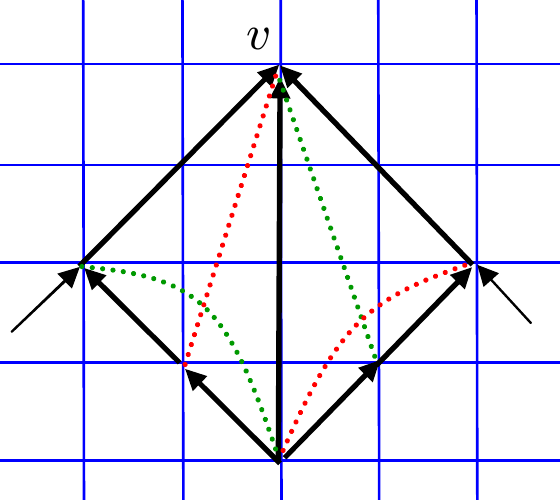}  
      \label{fig:conv-faces1}
  }
  \hspace{5mm}
\subfigure[ ] {    
    \includegraphics[scale=0.3]{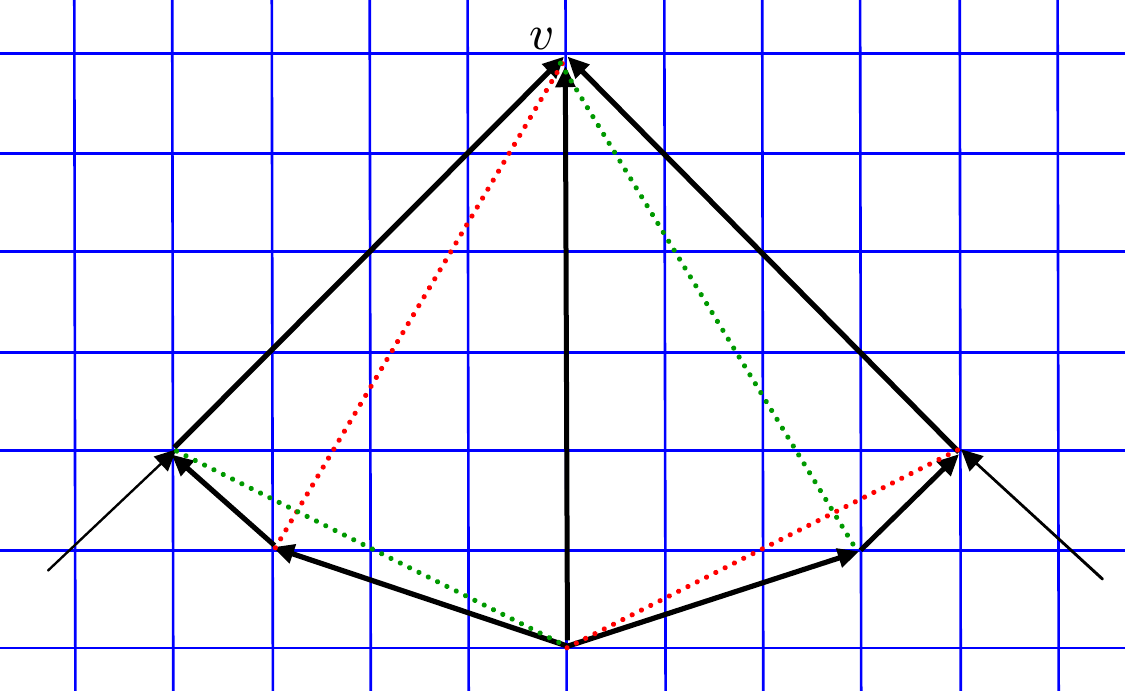}  
      \label{fig:conv-faces2}
  }
 \caption{Placement of vertex $v$ (a) in general and (b) for strict
 convexity.
  }
  \label{fig:faces}
\end{figure}

\section{Straight-line Biplanar Drawings of 1-Planar Graphs}
\label{sect:main}

It is clear that every 1-planar graph is biplanar. In fact, every
1-planar graph can be decomposed into a planar graph and a forest
\cite{a-n1pg-14}.  Towards a straight-line biplanar drawing, we
first consider triconnected 1-planar graphs with a pair of  edges
crossing in the outer face. These edges must be redrawn. We show
that there is a straight-line biplanar grid drawing, which is
1-planar with the exception of edges incident to the top vertices of
a W-configuration. In general,  there is a separation pair if there
is a pair of crossed edges in the outer face of a component. By
induction, every inner component at a separation pair has a
straight-line biplanar drawing, which is scaled down, so that it
fits into the first inner face of the outer component. So we obtain
a  biplanar drawing of a 1-planar graph. By upper and lower bounds
on the length of (segments of) edges, we can estimate that a single
inner component is scaled down at most by $c/n^5$ for some $c>0$. In
return,  the outer component is scaled up by $O(n^5)$. By induction
on the number of inner components and recursion on separation pairs,
the upscaling is bounded by $O(n^n)$, so that we use high precision
arithmetic with numbers with $O(n \log n)$ many digits.

\begin{lemma} \label{lem:3-connected-last}
Let $G$ be a  3-connected  1-planar graph   with a W-configuration
 in the outer face, so that the vertices $v_1$ and $v_2$ from the base
 are the  first   and  the vertices $t_1$ and $t_2$ on top are the last two vertices in
a canonical ordering of the planar skeleton of $G$. Then $G$ admits
a specialized straight-line  biplanar drawing $\Gamma(G)$ on a grid
of quadratic size.   A segment of an edge $e$ has length $\ell$ with
$1 \leq \ell \leq 8n-16$ in $\Gamma(G)$ if $e$ is uncrossed in (the
1-planar drawing of) $G$.
\end{lemma}

\begin{proof}
The outer face of the planar skeleton of $G$ is a quadrangle with
vertices  $v_1, v_2, t_1$ and $t_2$ from the W-configuration. Let
$t_2$ be the last vertex and assume that there is no inner
quadrangle containing $v_2$ and $t_2$. Otherwise, let $t_1$ be the
last vertex. The canonical ordering is $\{v_1,v_2\} < v < t_1 < t_2$
for any other vertex $v$. A bucket consists of one or two vertices,
since the faces of the planar skeleton are triangles or quadrangles.
There is a quadrangle if and only if it contains a pair of crossed
edges in the 1-planar drawing of $G$.

Our algorithm uses the shift method and incrementally adds a vertex
or the pair of vertices with a horizontal edge in between to an
intermediate drawing.
  Similar to Alam et al.~\cite{abk-sld3c-13}, we extend  the
algorithms by de Fraysseix et al.~\cite{fpp-hdpgg-90} and   Kant
\cite{k-dpguco-96}, so that there are strictly convex inner faces
for a planar drawing of the planar skeleton.
 Our algorithm is based on the algorithm for 3-connected planar
graphs by Kant \cite{k-dpguco-96}, which constructs a drawing with
convex inner faces on a grid of size  $(2n-4) \times (n-2)$.
Strict convexity for inner quadrangles is obtained by extra shifts,
so that there is a straight-line planar drawing of the planar
skeleton on a grid of size at most $(4n-8) \times (2n-4)$
\cite{abk-sld3c-13, br-scdpg-06}, see Fig.~\ref{fig:faces}.

In detail, we construct the biplanar drawing $\Gamma(G)$ as follows.
  The planar subgraph of $G- \{t_1, t_2\}$ is
drawn as described before with a horizontal line for
   $\edge{v_1}{v_2}$. All quadrangles containing
$t_2$ are prepared for strict convexity, so that edges $\edge{a}{b}$
and $\edge{b}{c}$ are not in a line  for a quadrangle $(a,b,c,
t_2)$.
   Crossed edges are inserted in the interior of the quadrangle
  from which they were removed if they are not incident to $t_1$ or $t_2$.
By Lemma~\ref{lem:draw-3conn}, the intermediate drawing is
straight-line and 1-planar on a  grid of size at most $(4n-16)\times
(2n-8)$. The uncrossed edges have length at least one and at most
$4n-16$.

By the restrictions for a canonical ordering \cite{k-dpguco-96}, the
last two vertices are placed one after another. First, add $t_1$ and
its incident edges to lower vertices including the crossed
 edge $\edge{t_1}{v_1}$, so that there is an isosceles right-angled
 triangle $(v_1,v_2,t_1)$ with $t_1$ on top. Edge  $\edge{t_1}{v_1}$ is colored red,
 so that it is transparent for black edges incident to $t_2$, see
 Fig.~\ref{fig:fpp-1}.
The lower neighbors of $t_2$ are on the contour of the drawing of
$G-t_2$ if edge $\edge{t_1}{v_1}$ is ignored. Let $v$ be the
rightmost lower neighbor of $t_2$ so that $v \neq v_2$.  If $v$ is
$x\geq 0$ units to the right of $t_1$, then shift $v_2$ by $2x+2$
extra units to the right and place $t_1$ at the intersection of the
diagonals through $v_1$ and $v_2$, so that it moves by  $(x+1, x+1)$
from its former position. Clearly, $2x+2 \leq 4n-8$. Thereafter, all
neighbors of $t_2$ are to the left of $t_1$, except for $v_2$, which
is the rightmost vertex.
  At last, shift $v_1$ one unit to the left and $v_2$ one unit to the right
  and place   $t_2$ at the intersection of the diagonals,
  so that $t_2$ is one unit vertically above
 $t_1$, and draw all edges incident to $t_2$ straight line, see Fig.~\ref{fig:fpp-2}.


All uncrossed edges of $G$  are colored black. Edge
$\edge{t_1}{v_1}$ is colored red. The edges incident to $t_2$ and
$v_2$ are colored black, so that all edges crossing an edge incident
 to $t_2$ or $v_2$ in the interior of a quadrangle
are colored red. Recall that there is no inner quadrangle with $t_2$
and $v_2$. For all other pairs of crossed edges in a quadrangle
choose any coloring, so that one crossed  edge is red and the other
is black.

By Lemma~\ref{lem:draw-3conn} and the placement of $t_1$, the
drawing of $G-t_2$ is straight-line 1-planar. The red edge
$\edge{t_1}{v_1}$ is the leftmost edge when $t_1$ is placed, so that
it does not cross any red edge. Hence, red edges do not cross, since
edges incident  to $t_2$ are black and other red edges are in the
interior of strictly convex quadrangles. Black edges do not cross if
they are not incident to $t_2$ by Lemma~\ref{lem:draw-3conn}. Edges
$\edge{t_2}{v_1}$ and $\edge{t_2}{v_2}$ are in the outer face and
are uncrossed.   Any other edge incident to $t_2$ is crossed by the
red edge $\edge{t_1}{v_1}$, which is not crossed by other edges, so
that $\edge{t_1}{v_1}$ is crossed by edges of a  fan at $t_2$. Edge
$\edge{t_2}{v}$ is crossed by an edge $\edge{u}{u'} \neq
\edge{t_1}{v_1}$ in $\Gamma(G)$ if and only if edges $\edge{t_2}{v}$
and $\edge{u}{u'}$ cross in the 1-planar drawing of $G$.
 Then $\edge{u}{u'}$ is colored red and there is a
strictly convex quadrangle $(t_2, u, v, u')$ in $\Gamma(G)$. Edge
$\edge{t_2}{v}$ is not crossed by any further edge, so that it is
crossed at most twice. Hence, $\Gamma(G)$   is a specialized
straight-line biplanar grid drawing.

The grid has size at most $(8n-16) \times (4n-8)$, since the size of
at most $(4n-8)\times (2n-4)$ by Lemma~\ref{lem:draw-3conn} is
extended by  at most $4n-8$ additional shifts for vertex $t_2$.
Clearly, edge $\edge{v_1}{v_2}$ is the longest edge in $\Gamma(G)$
with a length bounded by $8n-16$. Uncrossed edges of $G$, that is
edges of the planar skeleton, are uncrossed in $\Gamma(G)$ if they
are not incident to $t_2$, and uncrossed edges incident to $t_2$ are
only crossed by $\edge{t_1}{v_1}$ in $\Gamma(G)$. An uncrossed edge
has length at least one by Lemma~\ref{lem:draw-3conn}. If
$\edge{t_2}{v}$ is uncrossed in $G$, then it has two segments from
$t_2$ to the crossing point with $\edge{t_1}{v_1}$ and from there to
$v$, see Fig.~\ref{fig:fpp-3}. The first segment has length greater
than one, since $t_2$ is placed one unit above $t_1$, and the second
segment has length $\ell
> 1$, since $v$ is
  at least two units below $t_1$ and $\edge{t_1}{v_1}$ has slope
   $(h-1)/h$, where $h$ is the height (or half of the width) of the drawing.
  Hence, segments of uncrossed edges of $G$ have length at least
  one, whereas segments of crossed edges can be very short.
\end{proof}

\begin{figure}[t]
\centering
\subfigure[ ] {    
     \includegraphics[scale=0.3]{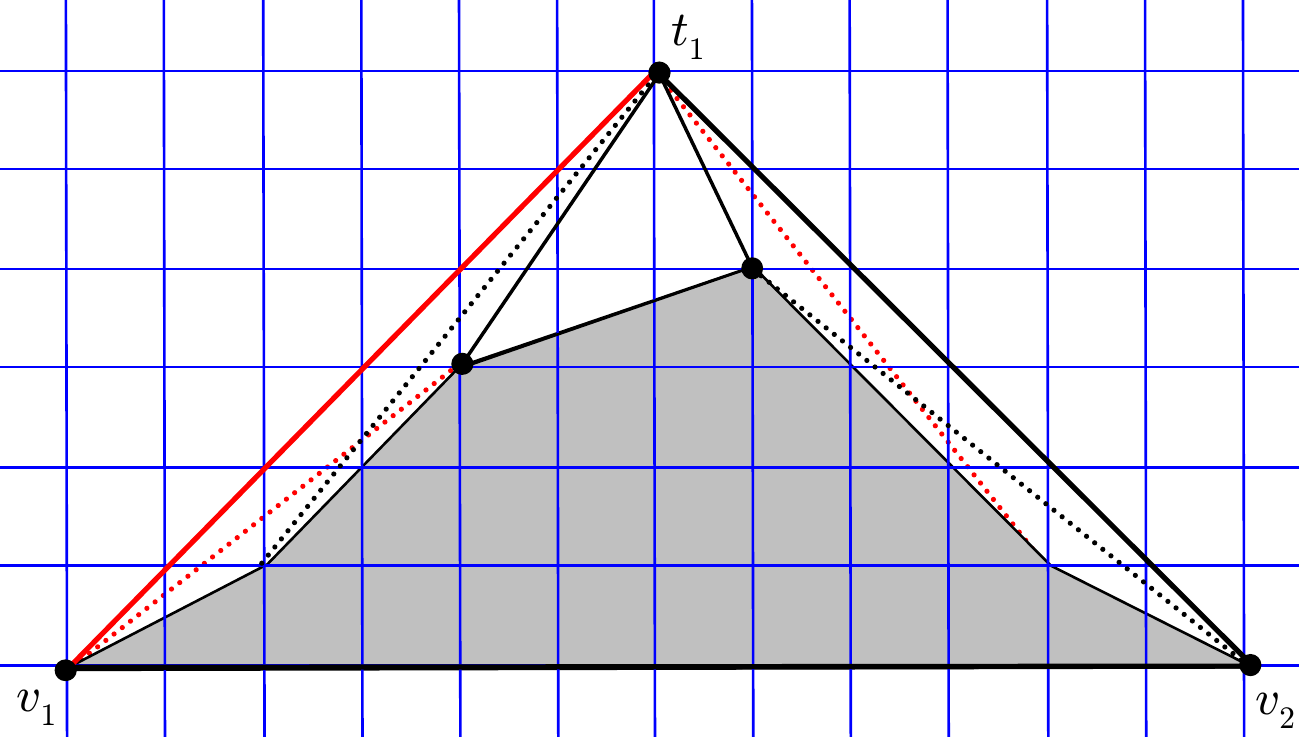}  
      \label{fig:fpp-1}
  }
\subfigure[ ] {    
     \includegraphics[scale=0.3]{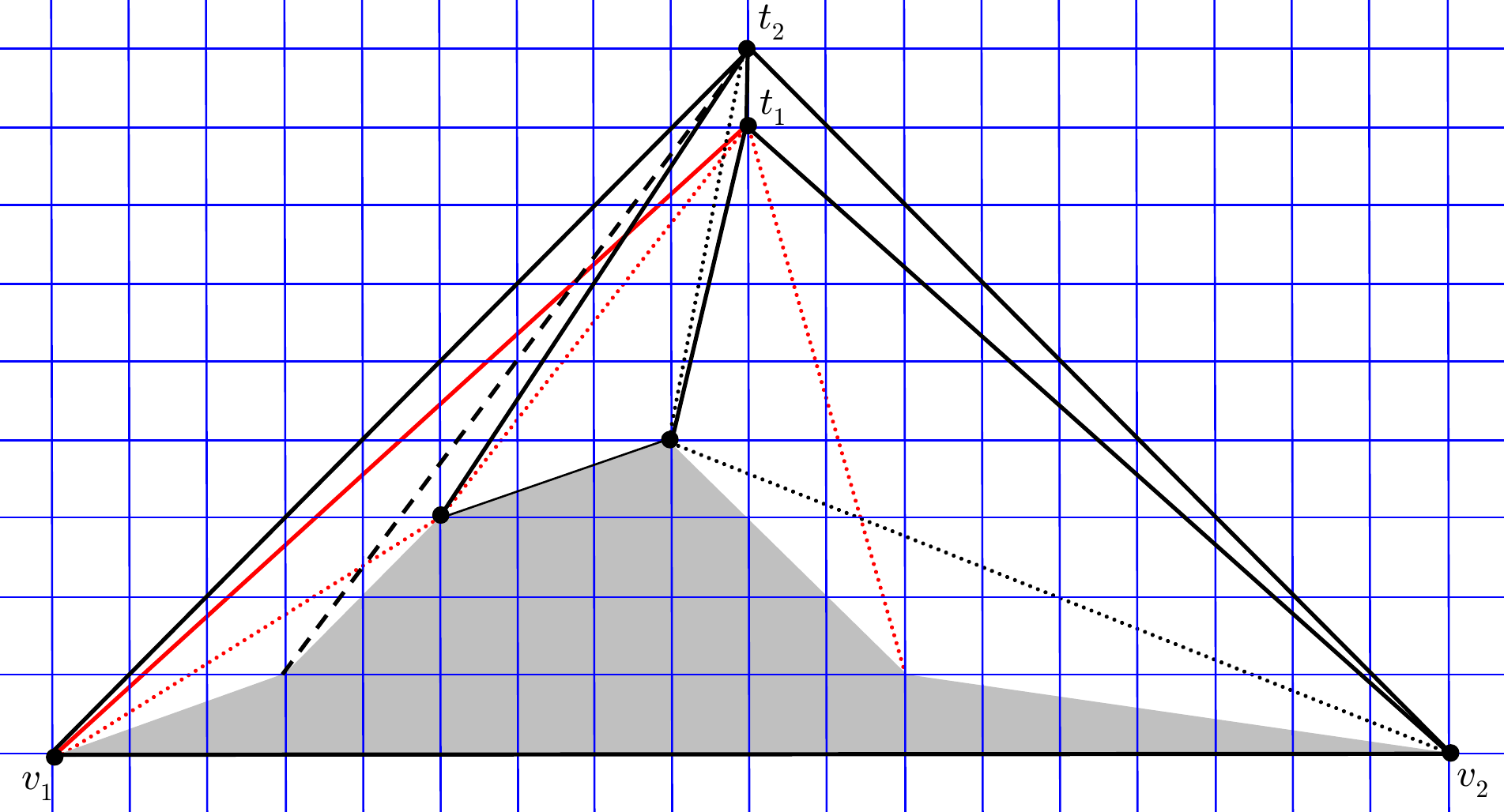}  
      \label{fig:fpp-2}
  }
\subfigure[ ] {    
     \includegraphics[scale=0.3]{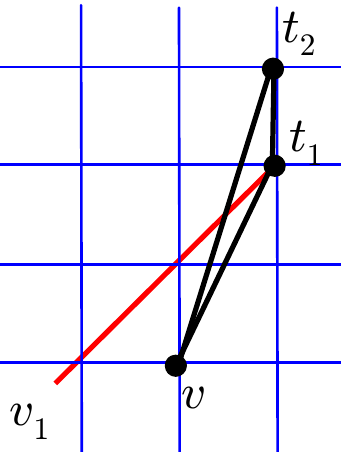}  
      \label{fig:fpp-3}
      }
 \caption{Illustration to the proof of Lemma~\ref{lem:3-connected-last}.
 (a) The drawing up to vertex $t_1$. (b) Extra shifts for $t_1$, the insertion of $t_2$
 and an edge coloring, and (c) the length of segments
  of edges crossed by $\edge{t_1}{v_1}$.
  }
  \label{fig:one-W-shifts}
\end{figure}

We now consider a W-configuration  with last vertex $t_1$, so that
the second top vertex $t_2$ cannot be the last but one vertex in the
canonical ordering. Then there are  separating triangles with
vertices $t_1$ and $t_2$. We can assume that a  flip of the
component at the base, as in Figs.~\ref{fig:Bconf} and
\ref{fig:Bconf-reverse}, does not help to avoid a separating
triangle with the top vertices. Otherwise,
Lemma~\ref{lem:draw-3conn} or Lemma~\ref{lem:3-connected-last} are
used for the changed embedding.

\begin{lemma} \label{lem:3-connected}
Every  3-connected  1-planar graph $G$ with a pair of crossed edges
  in the outer face admits a specialized straight-line biplanar
  drawing $\Gamma(G)$ on a grid of size $O(n^7)$. In addition,   the outer face is an isosceles
  obtuse angled triangle, and a segment of   edge $e$ in $\Gamma(G)$
  has length $\ell$  with $1 \leq \ell \leq c n^5$ for some $c>0$ if $e$ is uncrossed in
  (the 1-planar drawing of) $G$.
\end{lemma}

\begin{figure}[t]
\centering
\subfigure[ ] {    
    \includegraphics[scale=0.45]{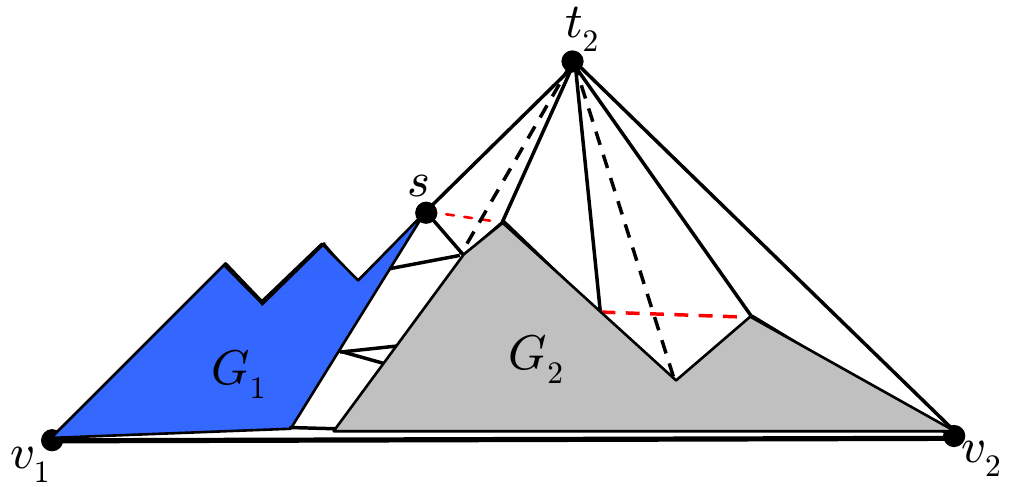}  
      \label{fig:uppergraph}
  }
  \hspace{1mm}
\subfigure[ ] {    
   \includegraphics[scale=0.45]{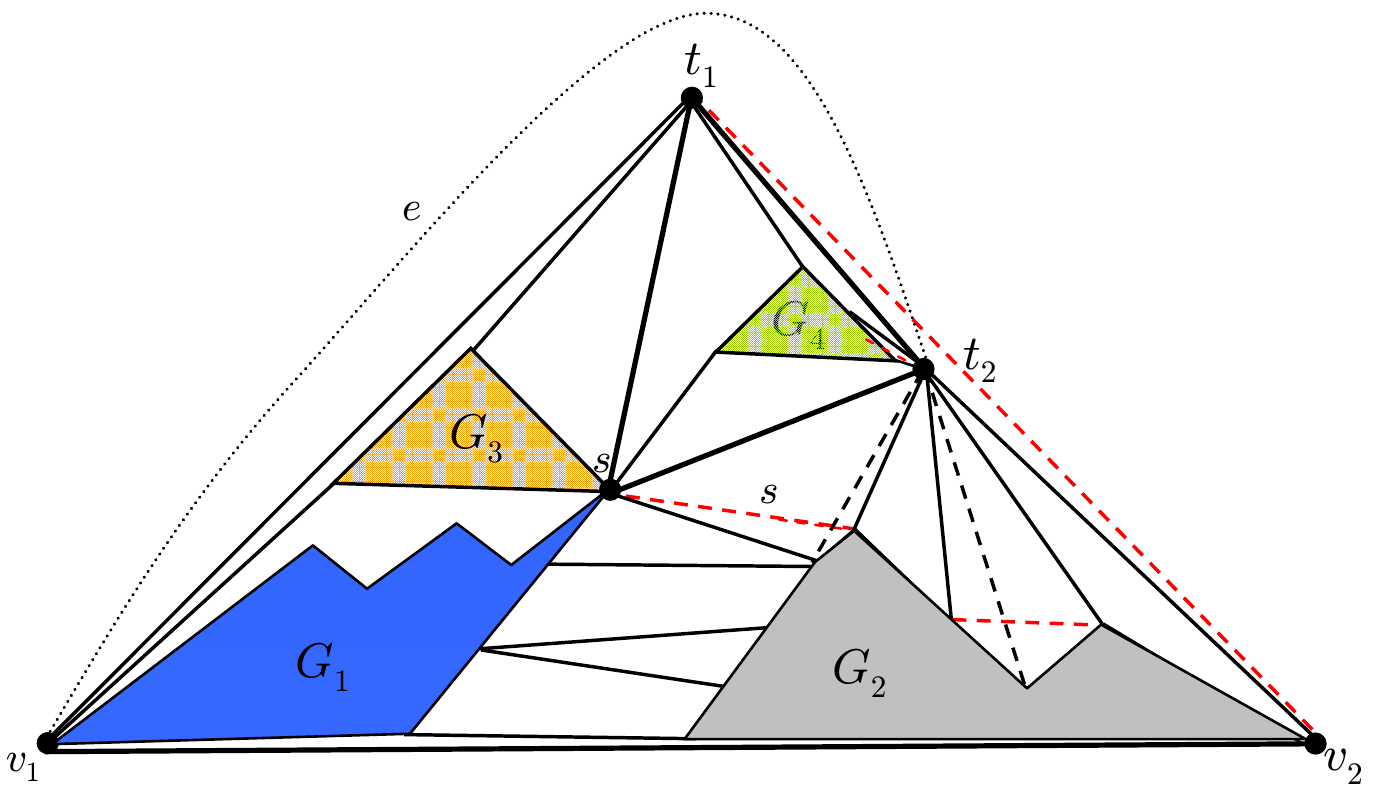}  
      \label{fig:gamma}
  }\\
  \subfigure[ ] {    
     \includegraphics[scale=0.45]{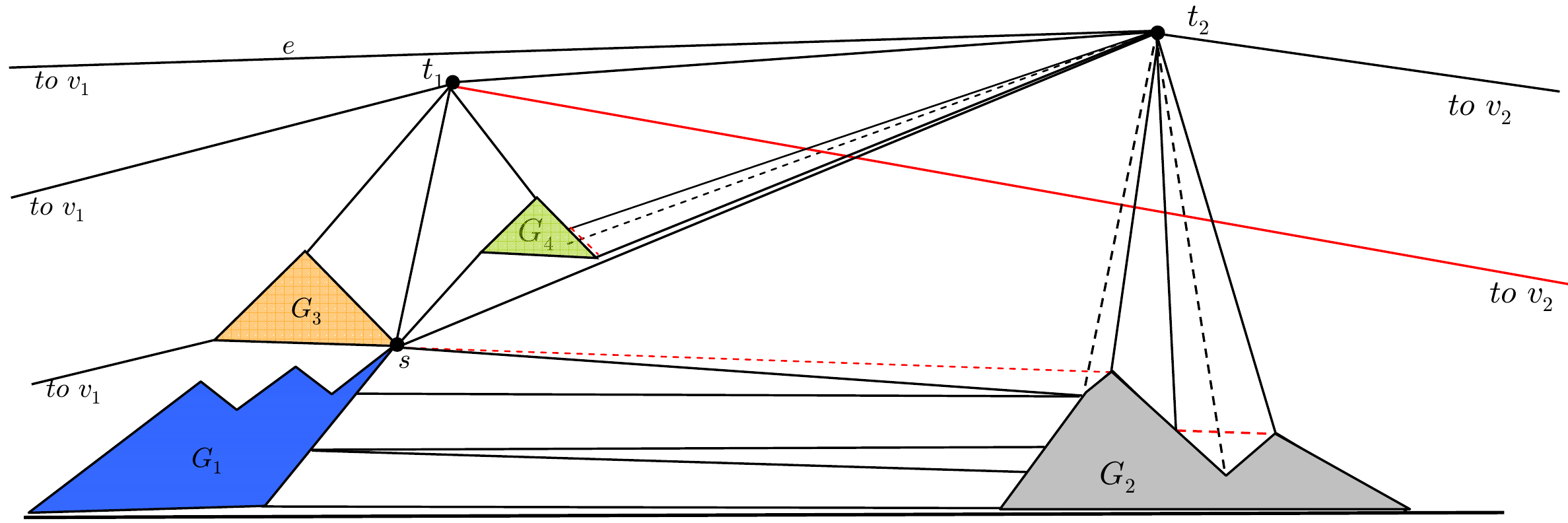}  
      \label{fig:gamma-prime}
  }
 \caption{Illustration to the proof of Lemma~\ref{lem:3-connected}.
 (a) A 1-planar drawing of $G$ up to vertex $t_2$, and
 (b) a 1-planar drawing of $G^-$, in which edge $e=\edge{v_1}{t_2}$, drawn black and dotted, is missing. Crossed edges are
 drawn dashed.
 (c) A biplanar drawing of $G$, which is obtained from the 1-planar
 drawing by  moving $t_2$  along
 edge $\edge{s}{t_2}$ above $t_1$ and shifting $G_2$.
  }
  \label{fig:lem-3connected}
\end{figure}

\begin{proof}
Consider a canonical ordering $v_1, v_2, \ldots,   t_1$ of the
planar skeleton of $G$, so that there is a W-configuration with base
$v_1, v_2$ and top vertices $t_1$ and $t_2$ in $G$  with $v_1$ left
of $v_2$. There is a separating triangle $(s,t_1, t_2)$ if $t_2$
cannot be the last but one vertex in a canonical ordering.  Choose
$s$ so that $(s,t_1, t_2)$
  is maximal in the sense that
 its vertices are not contained in another separating triangle
 with vertices $t_1$ and $t_2$. In other words,
  $s$ is the least  vertex in a separating triangle with $t_1$ and $t_2$.
 Vertices $s$ and $t_2$ are in the outer face of
the planar skeleton of $G -t_1$. Since there is a W-configuration,
we have $s \neq v_1$ and $s \neq v_2$.
 Assume $s < t_2$ in the canonical ordering. If this order is impossible,
 then there is a separating triangle $(s',t_1, t_2)$
   containing $s$, so that   $(s,t_1, t_2)$ is not maximal
   and $s$ is not the least such vertex, contradicting our assumption.
   Vertex $s$ is the leftmost
   lower neighbor of $t_2$, so that $t_2$ is first placed at the
   intersection of the diagonals through $s$ and $v_2$, as shown in
   Fig.~\ref{fig:uppergraph}.

Let $G^-$ be the subgraph of $G$ without the  crossed edge
$e=\edge{v_1}{t_2}$, drawn black and dotted in Fig.~\ref{fig:gamma}.
 Graph $G^-$ decomposes into four
subgraphs. Subgraph $G_4$ consist of the vertices in the interior of
the separating triangle $(s, t_1, t_2)$ including $s, t_1$ and
$t_2$, so that $s< t_2 < v < t_1$ for every vertex $v$ in the
separating triangle. Subgraph $G_3$ is induced by the vertices $v >
s$ and $v \not\in G_4$. Subgraph
 $G_2$ consists of the vertices $v < s$
 under  $t_2$, as described in \cite{ck-cgd-97, ck-mwgd-98,
cp-ltadpg-95, fpp-hdpgg-90, k-dpguco-96}, and $G_1$ is the remainder
with vertices $v < s$. If $s$ is to the left of $t_2$, then $G_1$ is
to the left of $G_2$  and $G_3$ is to the left of $G_4$. By the
shift method, the vertices of $G_2$ are shifted horizontally with
$t_2$.

We first construct a 1-planar drawing of $G^-$  as described  in
Lemma~\ref{lem:draw-3conn}, with the following modifications for
horizontal edges in $G_4$ and vertex $t_2$.
Vertex $t_2$ is first placed at the intersection of the diagonals
through $s$ and $v_2$, so that it is $h \geq 1$ units above $s$.
 Every vertex $v$ of  $G_4$ with $t_2 < v < t_1$ is placed above
the line $\edge{s}{t_2}$, using the shift method as described in
Lemma~\ref{lem:draw-3conn}, with the following exception. If there
is a    bucket in the canonical ordering 
with vertices $y$ and $z$,  so that $t_2$ is the rightmost lower
neighbor of $z$, then   there is a quadrangle $(x,y,z, t_2)$ in the
planar skeleton, so that $x$ is the leftmost lower neighbor of $y$
and $\edge{y}{t_2}$ is crossed by $\edge{x}{z}$ in $G$. The previous
algorithm \cite{k-dpguco-96} draws edge $\edge{y}{z}$ as a
horizontal line of length two. We draw it with slope $h/(h+2)$, so
that $(x,y,z,t_2)$ is a strictly convex quadrangle. This is achieved
by $2h$ extra shifts for $t_2$, so that $y$ is placed at $(h,h)$
from its former position.  Edges $\edge{x}{y}$ and $\edge{t_2}{z}$
have slope +1 and -1,
 respectively, when $y$ and $z$ are placed. The slope
 of $\edge{x}{t_2}$ is negative if $x \neq s$, since $x$  is placed
 after $t_2$ in this case. If there is a quadrangle $(s,y,z,t_2)$ just above
 $\edge{s}{t_2}$, then edge $\edge{s}{t_2}$ has slope $h/(3h+4)$ at
 the placement of $y$ and $z$, which is less than the slope of
 $h/(h+2)$ for $\edge{y}{z}$.
 Later on, edges are flattened by the shift method, which reduces the angular
 resolution, whereas a positive angle between adjacent adjacent edges is
 preserved.
If there is a triangle $(a,b,t_2)$ in the 1-planar drawing of $G_1$,
then $t_2$ is the rightmost lower neighbor of $b$, so that
$\edge{t_2}{b}$ has slope -1 when $b$ is placed. The slope of
$\edge{a}{b}$ is at least one if it is positive. It may be negative.

Since only techniques from the shift method \cite{fpp-hdpgg-90,
k-dpguco-96} are used, there is a straight-line strictly convex
drawing of the planar  skeleton of $G$ (or $G^-$) on a grid of size
at most $h(4n-8) \times h(2n-4)$, where $h$ is the vertical distance
between $s$ and $t_2$.
 In this drawing, edge $\edge{s}{t_2}$ has slope $\alpha$
 with $ \frac{h}{4hn_4} \leq \alpha < 1$,
 where $n_4$ is the number of vertices of $G_4$ in the interior of the triangle
 $(s, t_1, t_2)$, since $t_2$ is
 shifted horizontally at most $4h$ units per vertex of $G_4$
 according to Lemma~\ref{lem:draw-3conn}. The height  $h$ is bounded
 by  $1 \leq h \leq 2n_2-4$, since only vertices of $G_2$ shift
 $v_2$ relative to $s$, where $n_2$ is the set of vertices of $G_2$.
  The length of an
 uncrossed edge is at least one and at most $h(4n-8) < 8n^2$.
 The crossed edges are inserted into the quadrangles from which they
 were removed, so that there is a 1-planar drawing of $G^-$.

 This drawing is modified by moving $t_2$ above $t_1$, so that the missing  edge
 $\edge{v_1}{t_2}$  can be added for a
 specialized straight-line biplanar drawing $\Gamma(G)$.
 Note that $t_2$ has further neighbors only in $G_2$ and  $G_4$, since
 $(s, t_1,t_2)$ is a separating triangle and $s$ is the leftmost
 lower neighbor of $t_2$.
 Extend edge
 $\edge{s}{t_2}$ until it crosses the horizontal line
  two units above $t_1$, and move $t_2$ to the
    grid point at or  next right of the crossing point. Thereby,
 $t_2$ is moved $x_2$ units to the right and $y_2$ units upward.
 Edge $\edge{s}{t_2}$ is flattened, since $y_2 /x_2 \leq \alpha \leq y_2/(x_2-1)$.
 By the shift method, the vertices of $G_2$
 are shifted horizontally with $t_2$,  that is $x_2$ units to the
 right.
Next, reroute edge $\edge{t_2}{v_1}$ from $t_2$ through the grid
point one
  above $t_1$ until it crosses the x-axis, that is the horizontal
  line through $v_1$
and $v_2$, and  shift $v_1$  to the grid point at or just left of
the crossing point. Edge $\edge{t_1}{v_2}$ is rerouted and is drawn
along a ray through the grid point  one unit above $G_2$, which is
obtained by the tangent from $t_1$ to the convex hull of the drawing
of $G_2$. Then $v_2$ is shifted horizontally to the grid point at or
next right of the crossing point of the ray and the x-axis, see
Fig.~\ref{fig:gamma-prime}. Finally, shift $v_1$ or $v_2$
horizontally so that there is an isosceles triangle with top $t_2$.

Edge  $\edge{t_1}{v_2}$  is colored red. Edges incident to $t_2$ are
colored black, and their crossing edges are colored red. All
uncrossed edges and all crossed edges incident to $v_1$ are colored
black, whereas the coloring of the remaining pairs of crossed edges
can be chosen arbitrarily.
 This completes the construction of $\Gamma(G)$.

 We claim that $\Gamma(G)$
is a specialized straight-line biplanar grid drawing. Since there is
a straight-line 1-planar grid drawing for $G^-$, it suffices to
consider the modifications for $\Gamma(G)$. By construction, all
edges are drawn straight line  and all vertices are placed on grid
points.  Red edges do not cross, since they do not cross in the
1-planar drawing of $G^-$ and $\edge{t_1}{v_2}$ crosses only edges
  incident to $t_2$, which are black. Such an edge may be crossed by
  a red edge in the interior of a convex quadrangle.
 Two black edges do not cross if none of them is incident to
$t_2$, since they are drawn as in $\Gamma(G^-)$, or a vertex of the
edge (or both) is shifted horizontally if a vertex is in $G_2$.
Recall that horizontal shifts preserve planarity and convexity, as
observed at several places \cite{ck-cgd-97, fpp-hdpgg-90,
k-dpguco-96}. Clearly, two edges incident to $t_2$ do not cross,
since the edges are not drawn on top of one another. It remains to
consider a black edge  incident to $t_2$ and black edges of $G_2$
and $G_4$, see Fig.~\ref{fig:lem-3connected}.

Vertex $t_2$ and the vertices of $G_2$ are shifted $x_2$ units to
the right, and then $t_2$ is shifted $y_2$ units upward to its final
position. As observed by de Fraysseix et al.~\cite{fpp-hdpgg-90},
such shifts do not create crossings, and they preserve strict
convexity. In particular, a strict convexity of a quadrangle $(t_2,
s, u,u')$ with vertices $u, u' \in G_2$ is preserved, so that the
crossed edge $\edge{t_2}{u}$ is drawn straight-line in $\Gamma(G)$
and is crossed by the red edge $\edge{s}{u'}$. Hence, edges
$\edge{t_2}{v}$ and $\edge{u}{u'}$ with $u \in G_2$ cross in
$\Gamma(G)$ if and  only if they cross in in the 1-planar drawing of
$G^-$.

At last consider edges $\edge{v}{t_2}$ with $v \in G_4$ in the
1-planar drawing of $G^-$.  Vertex $v$ is placed above $t_2$ in
$\Gamma(G^-)$, but it is below $t_2$ in $\Gamma(G)$ after moving
$t_2$,  so that the slope of $\edge{v}{t_2}$  changes from negative
to positive, see Fig.~\ref{fig:lem-3connected}. We claim that moving
$t_2$ does not create new crossings with edges in $G_4$, since
$\edge{s}{t_2}$ is flat and horizontal edges in quadrangles with
$t_2$ are lifted. Consider the drawing of $G_4$ without vertex
$t_2$, which is obtained from the 1-planar drawing of $G^-$ by
removing $t_2$. The part of the outer face of $\Gamma(G - t_2)$ and
$\Gamma(G^- -t_2)$  between $s$ and $t_1$ consists of the neighbors
 $u_0,\ldots, u_r$ of $t_2$ in clockwise order, so that
 $s=u_0$ and $t_1=u_r$. These vertices are ordered in the canonical ordering, so that
 $u_i$ is placed after $u_{i-1}$.
 There are three cases:  (1) $u_{i-1}$ is the leftmost lower neighbor
 of $u_i$, (2) there is a quadrangle with vertices $u_{i-2}, u_{i-1}, u_i$
 and $t_2$, so that $\edge{u_{i-1}}{t_2}$ is crossed by $\edge{u_{i-2}}{u_i}$ or (3)
 $u_{i-1}$ is in the under-set of $u_i$.
In case (1), $u_i$ is placed  on the +1 diagonal through $u_{i-1}$,
so that $\edge{u_{i-1}}{u_i}$ has slope +1 when it is created. In
case  (2), edge $\edge{u_{i-1}}{u_i}$ is lifted and is not drawn as
a horizontal line, as in \cite{k-dpguco-96}. Edge
$\edge{u_i}{u_{i-2}}$ is red, since its crossing edge is incident to
$t_2$. Edge
 $\edge{u_{i-2}}{u_{i-1}}$ has slope +1 and
$\edge{u_{i-1}}{u_i}$ has slope $h/(h+2)$  when
  $u_{i-1}$ and $u_i$  are placed. At this moment, edge
$\edge{s}{t_2}$ has slope $h/W_i$ so that $W_i
> h+2$. In case (3), vertex $u_{i-1}$ is in the under-set of $u_i$,
so that (the slope of) edge $\edge{u_{i-1}}{u_i}$ is fixed, and is
only shifted from now on. Its slope is greater than one or may be
negative, which is even simpler, since $u_{i-1}$ is visible from
$t_2$ in $\Gamma(G)$ in the sense that there is an unobstructed
(uncrossed) line of sight between $u_{i-1}$  and $t_2$ if the red
edge $\edge{t_1}{v_2}$ is ignored.

Consider case (2) for the drawing of $G^-$. Cases (1) and (3) are
similar. Edge $\edge{u_{i-1}}{u_i}$ has slope $h/w_i$ and edge
$\edge{s}{t_2}$ has slope $h/W_i$ with $w_i < W_i$, when edge and
$\edge{u_{i-1}}{u_i}$ is drawn first.
Edge $\edge{u_{i-1}}{u_i}$ is flattened if vertices are placed above
it, so that $u_{i-1}$ and $u_i$ are shifted apart. Vertex $t_2$ is
shifted by the same quantity.   Hence, the slope of  $\edge{s}{t_2}$
remains less than the slope of $\edge{u_{i-1}}{u_i}$, since $h/w_i <
h/W_i$ implies 
$h/(w_i+x) < h/(W_i+x)$ if there is a horizontal shift by $x$ units.
In other words, edge $\edge{s}{t_2}$ in $\Gamma(G)$ is flatter that
any edge $\edge{u_{i-1}}{u_i}$ for
 with $1\leq i < r$. By induction, edge $\edge{u_i}{t_2}$ with $1\leq i<r$ is
 uncrossed in $\Gamma(G)$ if the red edges $\edge{t_1}{v_2}$  and
 $\edge{u_{i-1}}{u_{i+1}}$ in quadrangles $(u_{i-1}, u_i,
u_{i+1}, t_2)$ are ignored. In other words, moving $t_2$ in
direction of $\edge{s}{t_2}$ preserves the visibility of its
neighbors in $G_4$. Hence, the obtained drawing $\Gamma(G)$ is
biplanar.
 In addition,  the edges incident to $v_1$ are black, edge $\edge{t_1}{v_2}$ is
 crossed by edges of a fan at $t_2$, and  edges incident to
 $t_2$ are crossed twice if they are crossed in the 1-planar drawing of $G^-$.
Hence, the drawing is specialized.

Consider the size of the drawing. The 1-planar drawing of $G^-$ has
  size at most $h(4n-8) \times h(2n-4)$, where $h \leq 2n-4$ is the difference
  between $s$ and $t_2$ in y-dimension.
Edge $\edge{s}{t_2}$ has slope at least $1/4n$. Vertex $t_1$ is at
most $h(2n-4)$ units above $t_2$.
  Hence, $t_2$ is shifted at
most $16n^3$ units to the right. The ray from   $t_2$ along $t_1$
crosses the x-axis at most $128n^5$ units to the left of $v_1$,
since its slope is at least 
$1/16n^3$ and $t_1$ is at height at most $8n^2$. Similarly, the ray
from $t_1$ along the top of $G_3$ to $v_2$ crosses the x-axis at
most $32n^4$ units to the right of $t_1$. Hence, the width of the
drawing is  $O(n^5)$, and the height is   $O(n^2)$.

The upper bound on the length of an edge is $O(n^5)$. The lower
bound is one for an uncrossed edge, since there is a grid drawing.
An uncrossed edge of $G$ is crossed in $\Gamma(G)$ if it is incident
to $t_2$ and is crossed by $\edge{t_1}{v_2}$, see
Fig.~\ref{fig:gamma-prime}. The first segment from $t_2$ has length
at least two, since $t_2$ is two units above $t_1$  and the slope of
$\edge{t_1}{v_2}$  is negative. The crossing point of
$\edge{t_1}{v_2}$ and $\edge{t_2}{v}$ for a vertex in $G_2$ or $G_4$
is at least one unit above $v$, where $t_1$ is placed at least two
units above the vertices of $G_4$ in the drawing of $G^-$, so that
the second segment of $\edge{t_2}{v}$ has length at least one.
\end{proof}

If $\seppair{u}{v}$ is a separation pair with components $H_0,
\ldots, H_r$, then $\edge{u}{v}$ is an uncrossed edge if the graph
is planar-maximal 1-planar or it can be rerouted so that it is
uncrossed, see Fig.~\ref{fig:manycomponents}. There is a first inner
face next to $\edge{u}{v}$ for $H_i$   if the inner components $H_j$
with $j > i \geq 0$ are removed. Hence, $H_{i+1},\ldots, H_r$ are
drawn in the interior of the first inner face for $H_i$. This face
is a triangle whose third vertex is a crossing point or the first
vertex of $H_i^-$.    In a biplanar drawing, the triangle can be
intersected by a red edge from a previous component, for example if
$v$ is the top of the outer triangle, as stated in
Lemmas~\ref{lem:3-connected-last} and \ref{lem:3-connected} and
illustrated in Figs.~\ref{fig:fpp-2} and \ref{fig:gamma-prime}.
Since the last vertex can be  chosen among the two top vertices of a
W-configuration, for every inner component $H_i$, there is a
specialized biplanar drawing so that the edges incident to $v$ are
black.

\begin{figure}[t]
\centering
\subfigure[ ] {    
    \includegraphics[scale=0.6]{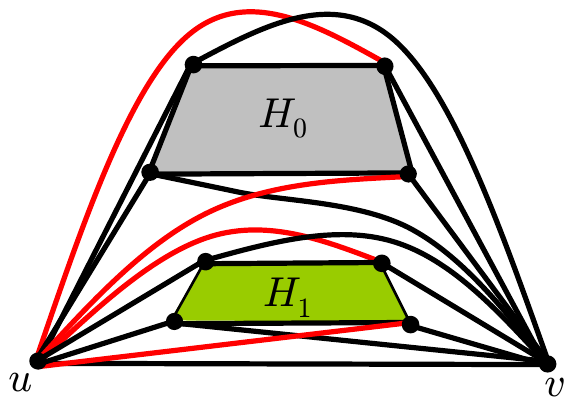}  
      \label{fig:comp1}
  }
  \hspace{2mm}
\subfigure[ ] {    
    \includegraphics[scale=0.7]{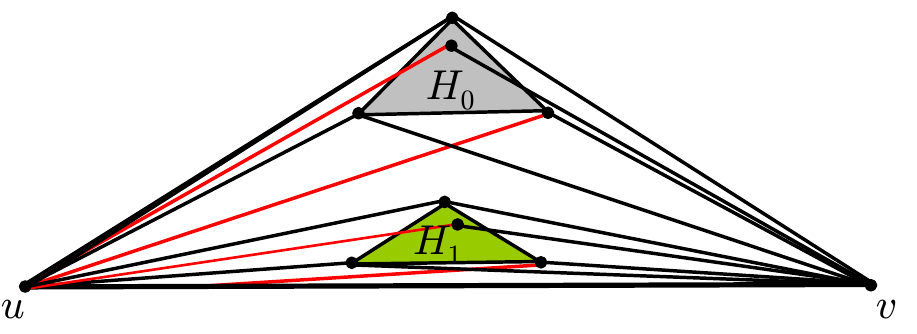}  
      \label{fig:comp2}
  }
 \caption{(a) A 1-planar drawing with W-configurations for the outer component $H_0$
 and inner
 component $H_1$.
   (b) A biplanar drawing with $H_1$   drawn in the first inner face of
   $H_0$. Note that quadrangles in the outer face of the planar
   skeleton of a 1-planar drawing are turned into a triangles in the
   biplanar drawing. Drawing (b) is scaled down.
  }
  \label{fig:manycomponents}
\end{figure}

The vertices of an uncrossed edge are a candidate for a separation
pair. In fact, there is a 1-planar drawing if every uncrossed edge
$e$ of a 1-planar drawing is substituted by a 1-planar graph $G_e$,
so that $e$ is  uncrossed  in $G_e$. For example, edge $e$ can be
substituted by $K_6$, drawn as a W-configuration. In general, this
property does not hold for the vertices of a crossed edge, since the
addition of a component violates 1-planarity or the edge is
uncrossed in another drawing.

\begin{lemma} \label{lem:all-components}
A 1-planar graph  admits a specialized  straight-line biplanar
drawing if it has a single separation pair with an inner and an
outer component. The drawing of the outer component is scaled by
$O(n^5)$ and segments of uncrossed edges have length at least one.
\end{lemma}

\begin{proof}

 Suppose that $\seppair{u}{v}$ is a separation pair so that
$G-\{u,v\}$ partitions into an outer component $H_0$ and an inner
component $H_1$. Edge $\edge{u}{v}$ is uncrossed in a 1-planar
drawing of $G$. It is the uncrossed base of $H_1$.  It is uncrossed
in a specialized straight-line biplanar drawing of $H_0$ if $u$ and
$v$ are not part of a W-configuration of $H_0$ by
Lemma~\ref{lem:draw-3conn}. Otherwise, there is a red edge
that crosses black edges incident to $v$ by
Lemmas~\ref{lem:3-connected-last} and~\ref{lem:3-connected}, see
Fig.~\ref{Fig-L3-H1}.

Independently, compute  specialized straight-line biplanar drawings
$\Gamma(H_0)$ and $\Gamma(H_1)$, which each include edge
$\edge{u}{v}$. Let $w_1$ and $h_1$ be the width and the height of
$\Gamma(H_1)$, so that
 $\edge{u}{v}$ is a horizontal line of length $w_1$.
In the drawing of $H_0$ it is a line of length $w_0$ with slope
$\alpha$. The first inner face of  $\Gamma(H_0)$ next to
$\edge{u}{v}$ is a triangle with height $h_0$. By
Lemmas~\ref{lem:3-connected-last} and ~\ref{lem:3-connected}, the
width and height of $\Gamma(H_0)$ and $\Gamma(H_1)$ is at most $128
n^5$ and $2 n^2$, respectively, and the first inner face has height
$h \geq 1/(256 n^5)$, since it is a triangle $(v_1, v_2, p)$, where
$\edge{v_1}{v_2}$ is a horizontal line of length at most $128n^5$
and $p$ is the third vertex one unit above $\edge{u}{v}$ or the
crossing point of edges with slope at least $1/(128n^5)$.

 Scale $\Gamma(H_1)$ in x-dimension by $w_0/ w_1$  and
 in y-dimension by $h_0/2h_1$. Then  $\Gamma(H_1)$ fits into
the first inner face of $\Gamma(H_0)$ if it is rotated by $\alpha$.
The line for edge $\edge{u}{v}$ coincides in both drawings and the
height of $\Gamma(H_1)$ is scaled down  to half of the height of the
first inner face. Clearly,
 edges of $H_0$ and $H_1$ do not cross. In reverse, the drawing of $H_0$ is
scaled by $O(n^5)$ if $H_1$ is drawn as described in
Lemma~\ref{lem:3-connected}.

If $\edge{u}{v}$ is crossed in $\Gamma(H_0)$, then there is a
W-configuration, so that one of $u$ and $v$ is the top of a
W-configuration by Lemmas \ref{lem:3-connected-last} and
\ref{lem:3-connected}. Assume that edges incident to $v$ are crossed
by a red edge $e$ in $\Gamma(H_0)$, that is $v=t_2$ in
Figs.~\ref{fig:fpp-3} and~\ref{fig:gamma-prime}. The edges incident
to $v$ are black in $\Gamma(H_0)$ and $\Gamma(H_1)$ by the
specialization.
Let $p$ be the crossing point of $e$ and $\edge{u}{v}$. First,
$\Gamma(H_1)$  is scaled down so that its base $\edge{u}{v}$ is
mapped onto the line between $u$ and a point  in the middle between
$u$ and $p$ in $\Gamma(H_0)$, see Fig.~\ref{fig:L3H0H1}. As before,
$\Gamma(H_1)$ is compressed, so that it fits into the first inner
face of $\Gamma(H_0)$, which is a quadrangle with a segment of $e$
in its boundary. In reverse, the scaling of $H_0$ is bounded by
$O(n^5)$. Thereafter,   apply the shift method to vertex $v$  in
$\Gamma(H_1)$ and move it to (the copy of) $v$ in $\Gamma(H_0)$.
Thereby, all edges incident to $v$ in $\Gamma(H_1)$ cross the red
edge $e$. The segment of an edge between $u$ and the crossing point
with $e$ has length at least one if $\Gamma(H_1)$ is unscaled, since
$v$ is at half the distance between $u$ and $p$ before its shift. By
the previous lemmas, every segment of an uncrossed edge of $G$ has
length at least one.

 Since $\Gamma(H_1)$ is specialized, the edges incident to
$v$ are black. Hence, the combined drawing of $H_1$ and $H_0$ is
specialized, straight-line and biplanar.
\end{proof}

\begin{figure}[t]
\centering
\subfigure[ ] {    
    \includegraphics[scale=0.7]{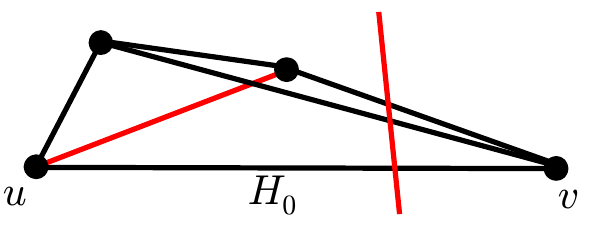}  
      \label{Fig-L3-H1}
  }
\subfigure[ ] {    
   \includegraphics[scale=0.7]{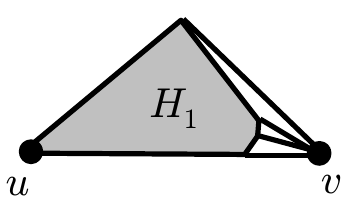}  
      \label{fig:L3H1}
  }
  \subfigure[ ] {    
     \includegraphics[scale=0.7]{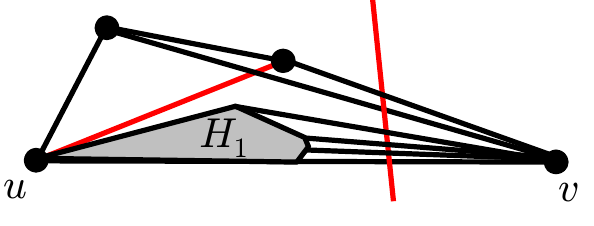}  
      \label{fig:L3H0H1}
  }
 \caption{Illustration to the proof of Lemma~\ref{lem:all-components}.
 (a) A  biplanar drawing of a part of $H_0$ with base $\{u,v\}$,
 (b) a specialized drawing of $H_1$ with base $\{u,v\}$, and
 (c) the composition of $H_0$ and $H_1$, so that $H_1$ is drawn in the first
 inner face of $H_0$.
  }
  \label{fig:Lemma3}
\end{figure}

We now proceed by induction on the number of inner components at
every separation pair, and by recursion on inner components at
separation pairs.

\begin{theorem} \label{thm:geothick-2}
 Every  1-planar graph has a tri-fan-crossing straight-line biplanar
 drawing. The drawing uses numbers with $O(n \log n)$ many digits and can be computed in linear time.
\end{theorem}
\begin{proof}
Assume that $G$ is planar-maximal 1-planar. Otherwise, use the
planarization and compute a planar-maximal augmentation  from a
1-planar drawing.
 Decompose $G$ into its 3-connected components and
store it in a decomposition tree, which may be an SPQR-tree
\cite{dett-gdavg-99}.  Let $H_0$ be the outer component which
remains from $G$ if all inner components at separation pairs are
removed. Then $H_0$ is triconnected. Compute a straight-line
specialized biplanar grid  drawing of $H_0$ according to
Lemmas~\ref{lem:draw-3conn}, \ref{lem:3-connected-last} or
\ref{lem:3-connected}.

Let $H_1, \ldots, H_k$ be inner components at a separation pair
$\seppair{u}{v}$ of $H_0$. Then $\edge{u}{v}$ is an uncrossed edge
of the 1-planar drawing of $H_0$. By induction, every $H_i$ admits a
specialized straight-line biplanar drawing with a horizontal line
for $\edge{u}{v}$. If $H_i$ has a pair of crossed edge in its outer
face, then choose the last vertex so that the edges incident to $v$
are black. By induction, there are  specialized biplanar drawings
$\Gamma(H_0,\ldots, H_i)$ and $\Gamma(H_{i+1})$, which can be
composed to a specialized biplanar drawing of $H_0, \ldots, H_{i+1}$
by Lemma~\ref{lem:all-components}. By induction there is a
tri-fan-crossing straight-line biplanar drawing for the inner
components at a separation pair of $H_0$.

Note that two separation pairs are independent in the sense  that
there are distinct first inner faces, so that the biplanar drawings
of components at two separation pairs do not interfere. By induction
and Lemma~\ref{lem:all-components}, there is a  tri-fan-crossing
straight-line biplanar drawing  $\Gamma(G)$.

A 1-planar graph with a single W-configuration can be drawn on a
grid of size $O(n^7)$ by Lemma~\ref{lem:3-connected}. For every
single inner component there is a scaling by  $O(n^5)$ by
Lemma~\ref{lem:all-components}. Otherwise, a scaling by $O(n^2)$
suffices. By induction and recursion this leads to a scaling of
  $O(n^n)$, since at least six vertices are necessary for a
W-configuration.  Hence, the coordinates of the drawing have  $O(n
\log n)$ many digits.

It takes linear time to compute $\Gamma(G)$ from a 1-planar drawing
of $G$, since every step can be done in linear time, such as the
planar-maximal augmentation, the SPQR-decomposition, the biplanar
drawing of single inner components, and the scaling by an affine
transformation.
\end{proof}

\begin{corollary} \label{cor:thickness-2}
Every 1-planar graph has geometric thickness at most two.
\end{corollary}

By a planarity test, we obtain:

\begin{corollary} \label{cor:compute-thickness}
The thickness and the geometric thickness of a 1-planar graph can be
computed in linear time.
\end{corollary}

\section{Conclusion}
We have shown that every 1-planar  graph admits a straight-line
biplanar drawing, that is it has geometric thickness two. The
following problems remain.

(1) Does every 1-planar graph admit a straight-line biplanar drawing
on a grid of polynomial size?


(2) Does every 1-planar graph admit  a rectangle visibility
representation? Rectangle visibility \cite{hsv-rstg-99} specializes
T-shape visibility \cite{b-Tshape-18}.

(3) What is the geometric (general, book) thickness of other beyond
planar graphs, for example $k$-planar, fan-crossing, fan-crossing
free and quasi-planar graphs \cite{dlm-survey-beyond-19}?

%
%
%
 \bibliographystyle{splncs04}
\bibliography{brandybibV8b.bib}

\begin{thebibliography}{10}
\expandafter\ifx\csname url\endcsname\relax
  \def\url#1{\texttt{#1}}\fi
\expandafter\ifx\csname urlprefix\endcsname\relax\def\urlprefix{URL }\fi
\expandafter\ifx\csname href\endcsname\relax
  \def\href#1#2{#2} \def\path#1{#1}\fi

\bibitem{sr-34}
E.~Steinitz, H.~Rademacher, {V}orlesungen {\"u}ber die {T}heorie der
  {P}olyeder, Julius Springer, Berlin, 1934.
\newblock \href {http://dx.doi.org/10.1007/978-3-642-65609-5}
  {\path{doi:10.1007/978-3-642-65609-5}}.

\bibitem{w-bv-36}
K.~Wagner, {B}emerkungen zum {V}ierfarbenproblem, {J}ahresbericht {D}eutsche
  {M}ath.-{V}ereinigung 46 (1936) 26--32.

\bibitem{fary-48}
I.~F\'{a}ry, On straight line representation of planar graphs, Acta Sci. Math.
  Szeged 11 (1948) 229--233.

\bibitem{s-cm-51}
S.~Stein, Convex maps, Proc. Amer. Math. Soc. 2 (1951) 464--466.
\newblock \href {http://dx.doi.org/10.2307/2031777}
  {\path{doi:10.2307/2031777}}.

\bibitem{con-dpgn-8585}
N.~Chiba, K.~Onoguchi, T.~Nishizeki, Drawing plane graphs nicely, Acta
  Informatica 22~(2) (1985) 187--201.
\newblock \href {http://dx.doi.org/10.1007/BF00264230}
  {\path{doi:10.1007/BF00264230}}.

\bibitem{cyn-lacdpg-84}
N.~Chiba, T.~Yamanouchi, T.~Nishizeki, Linear time algorithms for convex
  drawings of planar graphs, in: Progress in Graph Theory, Acadedmic Press,
  1984, pp. 153--173.

\bibitem{t-convex-60}
W.~T. Tutte, Convex representations of graphs, Proc. London Math. Soc 10 (1960)
  302--320.

\bibitem{t-hdg-63}
W.~T. Tutte, How to draw a graph, Proc. London Math. Soc. 13 (1963) 743--768.

\bibitem{fpp-hdpgg-90}
H.~de~Fraysseix, J.~Pach, R.~Pollack, How to draw a planar graph on a grid,
  Combinatorica 10 (1990) 41--51.
\newblock \href {http://dx.doi.org/10.1007/BF02122694}
  {\path{doi:10.1007/BF02122694}}.

\bibitem{s-epgg-90}
W.~Schnyder, Embedding planar graphs on the grid, in: {ACM-SIAM} Symposium on
  Discrete Algorithms, {SODA} 1990, {SIAM}, 1990, pp. 138--147.
\newblock \href {http://dx.doi.org/citation.cfm?id=320176.320191}
  {\path{doi:citation.cfm?id=320176.320191}}.

\bibitem{bfm-cd3cpg-07}
N.~Bonichon, S.~Felsner, M.~Mosbah, Convex drawings of 3-connected plane
  graphs, Algorithmica 47~(4) (2007) 399--420.
\newblock \href {http://dx.doi.org/10.1007/s00453-006-0177-6}
  {\path{doi:10.1007/s00453-006-0177-6}}.

\bibitem{ck-cgd-97}
M.~Chrobak, G.~Kant, Convex grid drawings of 3-connected planar graphs,
  Internat. J. Comput. Geom. Appl. 7~(3) (1997) 211--223.
\newblock \href {http://dx.doi.org/10.1142/S0218195997000144}
  {\path{doi:10.1142/S0218195997000144}}.

\bibitem{k-dpguco-96}
G.~Kant, Drawing planar graphs using the canonical ordering, Algorithmica 16
  (1996) 4--32.
\newblock \href {http://dx.doi.org/10.1007/BF02086606}
  {\path{doi:10.1007/BF02086606}}.

\bibitem{t-rdg-88}
C.~Thomassen, Rectilinear drawings of graphs, J. Graph Theor. 12~(3) (1988)
  335--341.
\newblock \href {http://dx.doi.org/10.1002/jgt.3190120306}
  {\path{doi:10.1002/jgt.3190120306}}.

\bibitem{help-ft1pg-12}
S.-H. Hong, P.~Eades, G.~Liotta, S.-H. Poon, F{\'a}ry's theorem for 1-planar
  graphs, in: J.~Gudmundsson, J.~Mestre, T.~Viglas (Eds.), {COCOON} 2012, Vol.
  7434 of {LNCS}, Springer, 2012, pp. 335--346.
\newblock \href {http://dx.doi.org/10.1007/978-3-642-32241-9\_29}
  {\path{doi:10.1007/978-3-642-32241-9\_29}}.

\bibitem{d-ds1pgd-13}
W.~Didimo, Density of straight-line 1-planar graph drawings, Inform. Process.
  Lett. 113~(7) (2013) 236--240.
\newblock \href {http://dx.doi.org/10.1016/j.ipl.2013.01.013}
  {\path{doi:10.1016/j.ipl.2013.01.013}}.

\bibitem{bsw-1og-84}
R.~Bodendiek, H.~Schumacher, K.~Wagner, {\"U}ber 1-optimale {G}raphen,
  Mathematische Nachrichten 117 (1984) 323--339.
\newblock \href {http://dx.doi.org/10.1002/mana.3211170125}
  {\path{doi:10.1002/mana.3211170125}}.

\bibitem{b-fan-20}
F.~J. Brandenburg, On fan-crossing graphs, Theor. Comput. Sci. 841 (2020)
  39--49.
\newblock \href {http://dx.doi.org/10.1016/j.tcs.2020.07.002}
  {\path{doi:10.1016/j.tcs.2020.07.002}}.

\bibitem{b-fcf-21}
F.~J. Brandenburg, Fan-crossing free graphs and their relationship to other
  classes of beyond-planar graphs, Theor. Comput. Sci. 867 (2021) 85--100.
\newblock \href {http://dx.doi.org/10.1016/j.tcs.2021.03.031}
  {\path{doi:10.1016/j.tcs.2021.03.031}}.

\bibitem{cpkk-fan-15}
O.~Cheong, S.~Har{-}Peled, H.~Kim, H.~Kim, On the number of edges of
  fan-crossing free graphs, Algorithmica 73~(4) (2015) 673--695.
\newblock \href {http://dx.doi.org/10.1007/s00453-014-9935-z}
  {\path{doi:10.1007/s00453-014-9935-z}}.

\bibitem{bkr-optimal2-17}
M.~A. Bekos, M.~Kaufmann, C.~N. Raftopoulou,
  \href{https://doi.org/10.4230/LIPIcs.SoCG.2017.16}{On optimal 2- and 3-planar
  graphs}, in: B.~Aronov, M.~J. Katz (Eds.), SoCG 2017, Vol.~77 of LIPIcs,
  Schloss Dagstuhl - Leibniz-Zentrum f{\"{u}}r Informatik, 2017, pp.
  16:1--16:16.
\newline\urlprefix\url{https://doi.org/10.4230/LIPIcs.SoCG.2017.16}

\bibitem{b-FOL-18}
F.~J. Brandenburg, A first order logic definition of beyond-planar graphs, J.
  Graph Algorithms Appl. 22~(1) (2018) 51--66.
\newblock \href {http://dx.doi.org/10.7155/jgaa.00455}
  {\path{doi:10.7155/jgaa.00455}}.

\bibitem{b-fan-fcf-18}
F.~J. Brandenburg, On fan-crossing and fan-crossing free graphs, Inf. Process.
  Lett. 138 (2018) 67--71.
\newblock \href {http://dx.doi.org/10.1016/j.ipl.2018.06.006}
  {\path{doi:10.1016/j.ipl.2018.06.006}}.

\bibitem{ku-dfang-14}
M.~Kaufmann, T.~Ueckerdt, The density of fan-planar graphs, Tech. Rep.
  arXiv:1403.6184 [cs.DM], Computing Research Repository ({CoRR}) (March 2014).

\bibitem{afps-grids-14}
E.~Ackerman, J.~Fox, J.~Pach, A.~Suk, On grids in topological graphs, Comput.
  Geom. 47~(7) (2014) 710--723.
\newblock \href {http://dx.doi.org/10.1016/j.comgeo.2014.02.003}
  {\path{doi:10.1016/j.comgeo.2014.02.003}}.

\bibitem{ppst-tgnlg-05}
J.~Pach, R.~Pinchasi, M.~Sharir, G.~T{\'{o}}th, Topological graphs with no
  large grids, Graphs and Combinatorics 21~(3) (2005) 355--364.
\newblock \href {http://dx.doi.org/10.1007/s00373-005-0616-1}
  {\path{doi:10.1007/s00373-005-0616-1}}.

\bibitem{t-thickness-63}
W.~T. Tutte, The thickness of a graph, Indag. Math. 25 (1963) 567--577.

\bibitem{mos-thickness-98}
P.~Mutzel, T.~Odenthal, M.~Scharbrodt, The thickness of graphs: {A} survey,
  Graphs and Combinatorics 14~(1) (1998) 59--73.
\newblock \href {http://dx.doi.org/10.1007/PL00007219}
  {\path{doi:10.1007/PL00007219}}.

\bibitem{aks-mgev-91}
A.~Aggarwal, M.~M. Klawe, P.~W. Shor, Multilayer grid embeddings for {VLSI},
  Algorithmica 6~(1) (1991) 129--151.
\newblock \href {http://dx.doi.org/10.1007/BF01759038}
  {\path{doi:10.1007/BF01759038}}.

\bibitem{nash-arboricity-61}
C.~S. J.~A. Nash-Williams, Edge-disjoint spanning trees of finite graphs,
  Journal of the London Mathematical Society 36~(1) (1961) 445--450.
\newblock \href {http://dx.doi.org/10.1112/jlms/s1-36.1.445}
  {\path{doi:10.1112/jlms/s1-36.1.445}}.

\bibitem{gw-ffg-92}
H.~N. Gabow, H.~H. Westermann, Forests, frames, and games: Algorithms for
  matroid sums and applications, Algorithmica 7 (1992) 465--497.
\newblock \href {http://dx.doi.org/10.1007/BF01758774}
  {\path{doi:10.1007/BF01758774}}.

\bibitem{k-tcg-73}
P.~Kainen, Thickness and coarseness of graphs, Abh. Math. Sem. Univ. Hamburg 39
  (1973) 88--95.

\bibitem{e-stgt-02}
D.~Eppstein, Separating thickness from geometric thickness, in: S.~G. Kobourov,
  M.~T. Goodrich (Eds.), {GD} 2002, Vol. 2528 of Lecture Notes in Computer
  Science, Springer, 2002, pp. 150--161.
\newblock \href {http://dx.doi.org/10.1007/3-540-36151-0\_15}
  {\path{doi:10.1007/3-540-36151-0\_15}}.

\bibitem{ag-thickness-76}
V.~Alekseev, V.~Gon\u{c}akov, The thickness of arbitrary complete graphs, Math.
  Sbornik 30~(2) (1976) 187--202.

\bibitem{deh-gtcg-00}
M.~B. Dillencourt, D.~Eppstein, D.~D. Hirschberg, Geometric thickness of
  complete graphs, J. Graph Algorithms Appl. 4~(3) (2000) 5--15.
\newblock \href {http://dx.doi.org/10.7155/jgaa.00023}
  {\path{doi:10.7155/jgaa.00023}}.

\bibitem{harary-61}
F.~Harary, Research problem, Bull. Amer. Math. Soc. 67~(542).

\bibitem{hsv-rstg-99}
J.~P. Hutchinson, T.~Shermer, A.~Vince, On representations of some
  thickness-two graphs, Computational Geometry 13 (1999) 161--171.
\newblock \href {http://dx.doi.org/10.1016/S0925-7721(99)00018-8}
  {\path{doi:10.1016/S0925-7721(99)00018-8}}.

\bibitem{y-epg4p-89}
M.~Yannakakis, Embedding planar graphs in four pages, J. Comput. System. Sci
  31~(1) (1989) 36--67.
\newblock \href {http://dx.doi.org/10.1016/0022-0000(89)90032-9}
  {\path{doi:10.1016/0022-0000(89)90032-9}}.

\bibitem{bkkpru-4pages-20}
M.~A. Bekos, M.~Kaufmann, F.~Klute, S.~Pupyrev, C.~N. Raftopoulou, T.~Ueckerdt,
  \href{https://journals.carleton.ca/jocg/index.php/jocg/article/view/504}{Four
  pages are indeed necessary for planar graphs}, J. Comput. Geom. 11~(1) (2020)
  332--353.
\newline\urlprefix\url{https://journals.carleton.ca/jocg/index.php/jocg/article/view/504}

\bibitem{y-4pages-20}
M.~Yannakakis, Planar graphs that need four pages, J. Comb. Theory, Ser. {B}
  145 (2020) 241--263.
\newblock \href {http://dx.doi.org/10.1016/j.jctb.2020.05.008}
  {\path{doi:10.1016/j.jctb.2020.05.008}}.

\bibitem{bbkr-book1p-17}
M.~A. Bekos, T.~Bruckdorfer, M.~Kaufmann, C.~N. Raftopoulou, The book thickness
  of 1-planar graphs is constant, Algorithmica 79~(2) (2017) 444--465.
\newblock \href {http://dx.doi.org/10.1007/s00453-016-0203-2}
  {\path{doi:10.1007/s00453-016-0203-2}}.

\bibitem{abk-sld3c-13}
M.~J. Alam, F.~J. Brandenburg, S.~G. Kobourov, Straight-line drawings of
  3-connected 1-planar graphs, in: S.~Wismath, A.~Wolff (Eds.), Proc. 21st {GD}
  2013, Vol. 8242 of {LNCS}, Springer, 2013, pp. 83--94.
\newblock \href {http://dx.doi.org/10.1007/978-3-319-03841-4\_8}
  {\path{doi:10.1007/978-3-319-03841-4\_8}}.

\bibitem{bsw-bs-83}
R.~Bodendiek, H.~Schumacher, K.~Wagner, {B}emerkungen zu einem
  {S}echsfarbenproblem von {G}. {R}ingel, Abh. aus dem Math. Seminar der Univ.
  Hamburg 53 (1983) 41--52.
\newblock \href {http://dx.doi.org/10.1002/mana.3211170125}
  {\path{doi:10.1002/mana.3211170125}}.

\bibitem{pt-gdfce-97}
J.~Pach, G.~T{\'o}th, Graphs drawn with a few crossings per edge, Combinatorica
  17 (1997) 427--439.
\newblock \href {http://dx.doi.org/10.1007/BF01215922}
  {\path{doi:10.1007/BF01215922}}.

\bibitem{br-scdpg-06}
I.~B{\'a}r{\'a}ny, G.~Rote, Strictly convex drawings of planar graphs,
  Documenta Mathematica 11 (2006) 369--391.

\bibitem{a-n1pg-14}
E.~Ackerman, A note on 1-planar graphs, Discrete Applied Mathematics 175 (2014)
  104--108.
\newblock \href {http://dx.doi.org/10.1016/j.dam.2014.05.025}
  {\path{doi:10.1016/j.dam.2014.05.025}}.

\bibitem{ck-mwgd-98}
M.~Chrobak, G.~Kant, S.~Nakano, Minimum-width grid drawings of plane graphs,
  Comput. Geomet. 11 (1998) 29--54.
\newblock \href {http://dx.doi.org/10.1016/S0925-7721(98)00016-9}
  {\path{doi:10.1016/S0925-7721(98)00016-9}}.

\bibitem{cp-ltadpg-95}
M.~Chrobak, T.~Payne, A linear-time algorithm for drawing a planar graph on a
  grid, Inform. Process. Lett. 54 (1995) 241--246.
\newblock \href {http://dx.doi.org/10.1016/0020-0190(95)00020-D}
  {\path{doi:10.1016/0020-0190(95)00020-D}}.

\bibitem{dett-gdavg-99}
G.~Di~Battista, P.~Eades, R.~Tamassia, I.~G. Tollis, Graph Drawing: Algorithms
  for the Visualization of Graphs, Prentice Hall, 1999.

\bibitem{b-Tshape-18}
F.~J. Brandenburg, \href{https://doi.org/10.1016/j.comgeo.2017.10.007}{T-shape
  visibility representations of 1-planar graphs}, Comput. Geom. 69 (2018)
  16--30.
\newblock \href {http://dx.doi.org/10.1016/j.comgeo.2017.10.007}
  {\path{doi:10.1016/j.comgeo.2017.10.007}}.
\newline\urlprefix\url{https://doi.org/10.1016/j.comgeo.2017.10.007}

\bibitem{dlm-survey-beyond-19}
W.~Didimo, G.~Liotta, F.~Montecchiani, A survey on graph drawing beyond
  planarity, {ACM} Comput. Surv. 52~(1) (2019) 4:1--4:37.
\newblock \href {http://dx.doi.org/10.1145/3301281}
  {\path{doi:10.1145/3301281}}.

\end{thebibliography}

\end{document}